\documentclass[a4paper]{article}
\usepackage{color}
\usepackage{amsmath}
\usepackage{amsthm}
\usepackage{amssymb}
\usepackage{amsmath}
\usepackage{mathcomp}
\usepackage{dsfont}

\newtheorem{definition}{Definition}[section]
\newtheorem{lemma}[definition]{Lemma}
\newtheorem{proposition}[definition]{Proposition}
\newtheorem{theorem}[definition]{Theorem}
\newtheorem{remark}[definition]{Remark}

\numberwithin{equation}{section}

\def\eps{\varepsilon}

\def\tr{\mathrm{tr}}

\def\bR{\mathbb{R}}
\def\R{\mathbb{R}}

\def\bR{\mathbb{R}}

\def\eps{\varepsilon}

\def\R{\mathbb{R}}
\let\e=\varepsilon

\let\.=\cdot
\let\0=\emptyset

\def\tr{\text{\rm tr}}

\def\square{\hbox{$\sqcap\kern-7pt\sqcup$}}

\def\be{\begin{equation}}
\def\ee{\end{equation}}
\def\bea{\begin{eqnarray}}
\def\eea{\end{eqnarray}}

\title{Mean-field evolution of fermions\\ with singular interaction}
\author{Chiara Saffirio   
\\ 
Institute of Mathematics, University of Zurich\\
Winterthurerstrasse 190, 8057 Zurich, Switzerland}

\begin{document}

\maketitle

\begin{abstract}
We consider a system of $N$ fermions in the mean-field regime interacting though an inverse power law potential $V(x)=|x|^{-\alpha}$, for $\alpha\in(0,1]$.
We prove the convergence of a solution of the many-body Schr\"{o}dinger equation to a solution of the time-dependent Hartree-Fock equation in the sense of reduced density matrices.
We stress the dependence on the singularity of the potential in the regularity of the initial data. The proof is an adaptation of \cite{PRSS}, where the case $\alpha=1$ is treated.
\end{abstract}

\section{Introduction}

{\it Fermionic mean-field regime.} We consider a system of $N$ particles obeying the Fermi statistics, whose state is represented by a wave function $\psi_N$ lying in $L^2_a(\R^{3N})$, the space of square integrable functions antisymmetric in the exchange of particles. The Hamiltonian of the $N$ particle system is given by
\be\label{eq:H}
H_N^{\rm ext} =\sum_{j=1}^N(-\e^2\Delta_{x_j}+V_{\rm ext}(x_j))+\frac{1}{N}\sum_{i<j}^N V(x_i-x_j),
\ee
where $V_{\rm ext}$ is an external potential confining the system in a volume of order one. As $V_{\rm ext}=0$,
the evolution in time of $\psi_N$ is given by  a solution to the Cauchy problem associated to the $N$-body Schr\"{o}dinger equation, here denoted by $\psi_{N,t}$: 
\be\label{eq:SE}
\left\{\begin{array}{l}
i\e\partial_t\psi_{N,t}=\left[\sum\limits_{j=1}^N(-\e^2\Delta_{x_j})+\frac{1}{N}\sum\limits_{i<j}^N V(x_i-x_j)\right] \psi_{N,t}\\\\
\psi_{N,0}=\psi_N\in L^2_a(\R^{3N}).
\end{array}
\right.
\ee

In \eqref{eq:SE}  the choice $\e=N^{-1/3}$ ensures the kinetic and the potential energy associated to \eqref{eq:H} to be of comparable order, namely $O(N)$. 
Observe that, in contrast with the bosonic case, the mean-field scaling for fermions comes coupled with a semiclassical limit (notice that here $\e$ plays the role of $\hbar$). This makes the analysis for fermions more complicated.  Different regimes have been considered in \cite{BBPPT}, \cite{BGGM}, \cite{FK} and \cite{PP}. 
\medskip

\noindent {\it Evolution of quasi-free states.} To begin with, it is convenient to introduce the one-particle reduced density matrix of a wave function $\psi_N$ as the non-negative trace class operator 
$$\gamma_N^{(1)}=N\tr_{2\dots N}|\psi_N \rangle\langle\psi_N |.$$

As we are interested in studying the dynamics of the system as $N\to\infty$, the choice of the initial data is crucial. Relevant initial states are given by the ground states of the Hamiltonian \eqref{eq:H}.
At zero temperature the equilibrium of confined states is approximated by Stater determinants, i.e.
\be\label{eq:slater}
\psi_{\rm Slater}(x_1,\dots,x_N)=\frac{1}{\sqrt{N!}}{\rm det}(f_i(x_j))_{1\leq i,j\leq N}
\ee
where $\{f_j\}_{j=1,\dots,N}$ is an orthonormal system in the one-particle space $L^2(\R^3)$. A Slater determinant is a quasi-free state completely determined by its one-particle reduced density matrix. 
A simple computation shows that the one-particle reduced density associated to \eqref{eq:slater} is given by
$$
\omega_{N}=\sum_{j=1}^N |f_j\rangle\langle f_j|,
$$
the orthogonal projection on the $N$ dimensional linear space ${\rm Span}\{f_1,\dots,f_N\}$. 
It minimizes the Hartree-Fock energy functional
\be\label{eq:E-HF}
\mathcal{E}_{\rm HF}(\omega)=\tr\ (-\e^2\Delta+V_{\rm ext})\,\omega+\frac{1}{2N}\iint V(x-y)\,[\omega(x;x)\,\omega(y;y)-|\omega(x;y)|^2]\,dx\,dy.
\ee
As proved in \cite{B}, \cite{GS}, the Hartree-Fock theory provides a good approximation of the ground state energy. It captures not only the leading order $O(N^{7/3})$ of the ground state as already established in the Thomas-Fermi theory (see \cite{L} and \cite{LSi} for a review on the subject), but also errors smaller than $O(N^{5/3})$.    
\smallskip

\noindent {\it Convergence towards the Hartree-Fock dynamics.}  
In \cite{EESY} and \cite{BPS13} it has been shown that the evolution of  a Slater determinant approximating the ground state of the Hamiltonian \eqref{eq:H} remains still close  to a Slater determinant. Its evolved one-particle reduced density matrix  is given by a solution to the time-dependent Hartree-Fock equation
\be
\label{eq:HF}
i\e\partial_t\omega_{N,t}=[-\e^2\Delta+(V*\rho_t)-X_t\,,\,\omega_{N,t}],
\ee
that is the Euler-Lagrange equation of \eqref{eq:E-HF}. For every $x\in\R^3$ 
$$\rho_t(x)=N^{-1}\omega_{N,t}(x;x)$$
 is the density associated to the one-particle reduced density matrix $\omega_{N,t}$, $(~V~*~\rho_t~)$ represents the so-called direct term while $X_t$ is the exchange term defined through the operator kernel
$$X_t(x;y)=\frac{1}{N}V(x-y)\omega_{N,t}(x;y).$$
More precisely, in \cite{BPS13} it has been proved that the Hartree-Fock approximation holds for initial states close to a Slater determinant with a semiclassical structure, namely
$$\omega_{N,t}(x;y)\simeq N\,\varphi\left(\frac{x-y}{\e}\right)\,\psi\left(\frac{x+y}{2}\right),$$
where $\psi$ and $\varphi$ determine respectively the density of particles and the momentum distribution. Heuristically the integral kernel of $\omega_{N,t}$ varies in the direction $x-y$ on scales  $O(\e)$ and in the direction $x+y$ on scales $O(1)$. This is precisely the  structure that is expected to hold true in Slater determinants approximating equilibrium states. This reflects on the following structure
\be
\label{eq:bounds-semiclassical}
\tr\ |[x,\omega_N]|\leq CN\e\quad \quad \tr\ |[-i\e\nabla,\omega_N]|\leq CN\e.
\ee   
In \cite{BPS13} the propagation (global) in time of the bounds \eqref{eq:bounds-semiclassical} has been shown, allowing for an approximation of the many-body Schr\"{o}dinger equation by the Hartree-Fock dynamics on time scales of order one.

The extension of this result to mixed states has been proved in \cite{BJPSS}.
\medskip

\noindent {\it Mean-field in presence of singular interactions.}
When dealing with singular interactions $V(x)={1}/{|x|^\alpha}$, $\alpha\in(0,1]$, the Hamiltonian takes the form
\be\label{eq:H-sing}
H_N=\sum_{i=1}^N(-\e^2\Delta_{x_i})+\frac{1}{N}\sum_{i<j}^N\frac{1}{|x-y|^{\alpha}}.
\ee
In particular, the case $\alpha=1$  treated in \cite{PRSS} represents a system of $N$ fermions interacting through a Coulomb potential, which describes for instance the dynamics of large atoms and molecules. In this case, the choice $\e=N^{-1/3}$ is justified by a rescaling of the space variables at a scale $O(N^{-1/3})$ (the typical distance of electrons from the nucleus) as suggested by the Thomas-Fermi theory (see \cite{L}, \cite{LSi}). An analogue reasoning applies to the case of inverse power law potentials and, by appropriately scaling the time variable,  it leads to
\be\label{eq:SE-sing}
i\e\partial_t\psi_{N,t}=\sum_{i=1}^N(-\e^2\Delta_{x_i})+\frac{1}{N}\sum_{i<j}^N\frac{1}{|x-y|^{\alpha}}.
\ee   
More details on the rigorous justification of the mean-field scaling coupled to the semicalssical one in the Coulomb case can be found in \cite{PRSS}.   

\begin{theorem}\label{thm:main}
Let $\omega_N$ be a sequence of orthogonal projections on $L^2 (\bR^3)$, with $\tr \, \omega_N = N$.
Let $\omega_{N,t}$ the solution of the Hartree-Fock equation (\ref{eq:HF}) with initial data $\omega_{N,0} = \omega_N$.  We assume that 
\begin{itemize}
\item[i)]
$\tr\, (-\eps^2 \Delta) \, \omega_N \leq C N$, for a constant $C >0$ independent of $N$; 
\item[ii)] there exist $T > 0$, $p > 6/(3-2\alpha-6\delta)$ for $\delta\in(0,\frac{1}{2})$ and  $C > 0$ such that 
\begin{equation}\label{eq:ass-main} \sup_{t \in [0;T]}  \, \sum_{i=1}^3 \left[ \| \rho_{|[x_i,\omega_{N,t}]|} \|_1 + \|  \rho_{|[x_i,\omega_{N,t}]|} \|_p \right] \leq C N \eps \,,  \end{equation}
where  $\rho_{|[x_i,\omega_{N,t}]|}(z)=|[x_i,\omega_{N,t}]|(z;z)$ is the function obtained by considering the diagonal kernel of the operator $|[x_i,\omega_{N,t}]|$.
\end{itemize}

Let $\psi_N \in L^2_a (\bR^{3N})$ be such that its one-particle reduced density matrix $\gamma_{N}^{(1)}$ satisfies
\begin{equation}\label{eq:conden} \tr \, \left| \gamma^{(1)}_N - \omega_N \right| \leq C N^\beta \end{equation}
for a constant $C > 0$ and an exponent $0 \leq \beta < 1$.

Consider the evolution $\psi_{N,t}= e^{-iH_N t/\eps} \psi_N$, with the Hamiltonian \eqref{eq:H-sing} and let  
$\gamma^{(1)}_{N,t}$ be the corresponding one-particle reduced density matrix. Then, for every $\eta > 0$, there exists $C >0$ such that 
\begin{equation}\label{eq:HS-bd} 
\sup_{t \in [0;T]} \, \left\| \gamma_{N,t}^{(1)} - \omega_{N,t} \right\|_\text{HS} \leq C \left[ N^{\beta/2} + N^{(3-2\alpha)/2(3-\alpha) + \eta}  \right] 
\end{equation}
and 
\begin{equation}\label{eq:tr-bd} 
\sup_{t \in [0;T]} \, \tr \left| \gamma_{N,t}^{(1)} - \omega_{N,t} \right| \leq C \left[ N^{(1+\beta)/2} + N^{(3-2\alpha)/2(3-\alpha) + \eta}\right] . 
\end{equation}  
\end{theorem}
\begin{remark}
Some remarks are in order.
\begin{enumerate}
\item The bounds \eqref{eq:HS-bd} and \eqref{eq:tr-bd} guarantee that the one-particle density matrix $\gamma_{N,t}^{(1)}$ is close to the solution of the Hartree-Fock equation in Hilbert-Schmidt and in trace norms as $N$ is sufficiently large. Indeed, since
 $\| \omega_{N,t} \|_\text{HS} = N^{1/2}$, $\| \gamma_{N,t}^{(1)} \|_\text{HS} = N^{1/2}$  and $\tr\, \omega_{N,t} = N$, $\tr\, \gamma_{N,t}^{(1)} = N$ , Eq.ns \eqref{eq:HS-bd} and \eqref{eq:tr-bd} asserts that the difference between $\gamma_{N,t}^{(1)}$ and $\omega_{N,t}$ is smaller than the size of each component, thus the Hartree-Fock equation is a good approximation for the many-body evolution with singular interaction. 
\item The exponent $0 \leq \beta < 1$ measures the initial number of excitations. In other words, $\beta$ measures the number of particles that are not in the Slater determinant when $t=0$.
\item The bounds are $N=\e^{-3}$ dependent. As already pointed out, it encodes the fact that the mean-field scaling is linked to a semicalssical limit in the fermionic setting. More precisely, as $N\to\infty$, the Wigner transform of a solution to the Hartree-Fock equation \eqref{eq:HF} converges (weakly) to a solution to the Vlasov equation. This statement has been proved in the case of regular interactions for pure states in \cite{BPSS}. Several results are available  in the case of mixed states, see for instance \cite{NS}, \cite{Sp} \cite{AKN},  \cite{GolsePaul}, \cite{APPP}, \cite{BGGM}. The papers \cite{MM}, \cite{LP} and \cite{FLP} deal with singular interactions (here included the Coulomb case) for mixed states. For pure states, the semiclassical limit towards the Vlasov equation with singular interaction potential is an open problem. 
\item  Assumption $ii)$ is very strong. Indeed, we ask the bound to hold true for $\omega_{N,t}$ as $t\geq 0$. However, the problem of the derivation of the Hartree-Fock equation from a system of many interacting fermions is reduced to a PDE problem. Namely, it remains to prove that the bound in assumption $ii)$, if assumed at time t=0, propagates at positive times. There is one special situation in which assumption $ii)$ is satisfied without any further requirement: the case of translation invariant Slater determinants. This is for instance the case of a system of $N$ fermions in a finite box of order one with periodic boundary conditions. In such a setting, the Hartree-Fock dynamics becomes trivial. Nevertheless it is interesting that the non trivial dynamics given by the $N$-body Schr\"{o}dinger equation can be approximated in the mean-field limit by a trivial one.
\end{enumerate}
\end{remark}

%%%%%%%%%%%%%%%%%%%%%%%%%%%%%%%%%%%%%%%%%%%%%
%%%%%%%%%%%%%%%%%%%%%%%%%%%%%%%%%%%%%%%%%%%%%
\section{Second quantization formalism}

Let us introduce the formalism of second quantization. We consider over $L^2(\R^3)$ the fermionic Fock space 
$$\mathcal{F}=\bigoplus_{n\geq 0} L^2_a(\R^{3n}).$$
On $\mathcal{F}$ we define for every $f\in L^2(\R^3)$ the creation and annihilation operators in terms of operator valued distributions $a_x^*,\ a_x$, $x\in\R^3$,
$$a^*(f)=\int f(x)\,a_x^*\,dx,\quad a(f)=\int \bar{f}(x)\,a_x\,dx.$$ 
The second quantization of an operator $O$ on $L^2(\R^3)$ is defined as
$$d\Gamma(O)=\int O(x;y)\,a_x^*\,a_y\,dx\,dy,$$
where $O(x;y)$ is the integral kernel of $O$. In particular $O={\bf 1}$ corresponds to the second quantization of the number of particle operator $\mathcal{N}$
$$\mathcal{N}=\int a_x^*\,a_x\,dx.$$ 

Given a vector in the Fock space $\Psi\in\mathcal{F}$, the one-particle reduced density matrix associated with $\Psi$ is the non-negative trace class operator on $L^2(\R^3)$ whose kernel is given by
\be\label{eq:1pdm}
\gamma_{\Psi}^{(1)}(x;y)=\langle \Psi,a_y^*\,a_x\,\Psi \rangle.
\ee 
Moreover, given $\Psi\in\mathcal{F}$ and a one-particle operator $O$ on $L^2(\R^3)$, the expectation of the second quantization of $O$ in the state $\Psi$ is given by
$$
\langle \Psi,d\Gamma(O)\,\Psi \rangle=\int O(x;y)\,\langle \Psi,a_x^*\,a_y\,\Psi\rangle\,dx\,dy=\tr\ O\,\gamma_{\Psi}^{(1)}.
$$ 
Notice that, according to this definition, the trace of the one-particle reduced density matrix \eqref{eq:1pdm} represents the expected number of particles in the state $\Psi$, namely
$$
\tr\ \gamma_{\Psi}^{(1)}=\langle \Psi,\mathcal{N}\Psi \rangle.
$$ 

The next Lemma is taken from \cite{BPS13} and collects some useful bounds. 

\begin{lemma}\label{lemma:2quant}
Let $O$ be a bounded operator on $L^2(\R^3)$. Then, for every $\Psi\in\mathcal{F}$,
\begin{equation*}\begin{split}
\langle \Psi,d\Gamma(O)\,\Psi \rangle&\leq \|O\|\langle \Psi,\mathcal{N}\Psi\rangle,\\
\|d\Gamma(O)\Psi\|&\leq \|O\|\,\|\mathcal{N}\Psi\|.
\end{split}
\end{equation*}
Moreover, if $O$ is a Hilbert-Schmidt operator, 
\begin{equation*}
\begin{split}
\|d\Gamma(O)\Psi\|&\leq\|O\|_{\rm HS}\|\mathcal{N}^{1/2}\Psi\|,\\
\left\| \int O(x;y)\,a_x\,a_y\,\Psi\,dx\,dy\,\right\|&\leq\|O\|_{\rm HS}\|\mathcal{N}^{1/2}\Psi\|,\\
\left\|\int O(x;y)\,a_x^*\,a^*_y\,\Psi\,dx\,dy\,\right\|&\leq\|O\|_{\rm HS}\|\mathcal{N}^{1/2}\Psi\|.
\end{split}
\end{equation*}
If $O$ is a trace class operator, we have
\begin{equation*}
\begin{split}
\|d\Gamma(O)\Psi\|&\leq 2\|O\|_{\tr },\\
\left\| \int O(x;y)\,a_x\,a_y\,\Psi\,dx\,dy\,\right\|&\leq 2\tr\ |O|,\\
\left\|\int O(x;y)\,a_x^*\,a^*_y\,\Psi\,dx\,dy\,\right\|&\leq 2\tr\ |O|.
\end{split}
\end{equation*}
\end{lemma}

We introduce the second quantization of the Hamiltonian \eqref{eq:H-sing} as the self-adjoint operator $\mathcal{H}_N$ whose restriction to the $n$-particle sector of the Fock space $\mathcal{F}$ is 
\be\label{eq:H-FS}
\mathcal{H}_N|_{\mathcal{F}_n}=\sum_{j=1}^n-\e^2\Delta_{x_j}+\frac{1}{N}\sum_{i<j}^n \frac{1}{|x_i-x_j|^{\alpha}}.
\ee
In particular, observe that the Hamiltonian $\mathcal{H}_N$ coincides with the Hamiltonian \eqref{eq:H-sing} when restricted to $\mathcal{F}_N$, the $N$-particle sector of the Fock space $\mathcal{F}$. Therefore, the dynamics of an initial data in $\mathcal{F}$ with $N$ particles coincides with the evolution given by \eqref{eq:SE}, where $V(x)=1/|x|^\alpha$.
%$$
%\mathcal{H}_N=\e^2\int \nabla_x a_x^*\nabla_x a_x\,dx+\frac{1}{N}\int V(x-y)\,a_x^*a_y^*a_ya_x\,dx\,dy. 
%$$

To define the second quantization of a Slater determinant \eqref{eq:slater}, we introduce $\{f_j\}_{j=1,\dots,N}$, an orthonormal system in $L^2(\R^3)$, and the vacuum $\Omega$. A Slater determinant on the Fock space $\mathcal{F}$ is given by
\be\label{eq:slater-2q}
a^*(f_1)\dots a^*(f_N)\Omega=\left\{ 0,\dots,0,\frac{1}{\sqrt{N!}}{\rm det}(f_i(x_j))_{1\leq i,j\leq N},0,\dots,0 \right\}.
\ee
These states enjoy a remarkable structure: they can be obtained as the action of a Bogoliubov transformation on the vacuum $\Omega$. More precisely, in our context a fermionic Bogoliubov trasformation is a unitary linear map 
$$\omega:L^2(\R^3)\oplus L^2(\R^3)\to L^2(\R^3)\oplus L^2(\R^3)$$ 
of the form
\begin{equation*}
\omega=\left(\begin{array}{ll}
u &\quad \bar{v}\\
v &\quad \bar{u}
\end{array}
\right)
\end{equation*} 
where $u,\,v:L^2(\R^3)\to L^2(\R^3)$ are linear maps such that 
\begin{equation*}
\begin{split}
& u^*u+v^*v=1,\\
& u^*\bar{v}+v^*\bar{u}=0.
\end{split}
\end{equation*}
We say that the Bogoliubov transformation $\omega$ is implementable on the fermionic Fock space $\mathcal{F}$ if there exists a unitary operator $\mathfrak{R}_\omega:\mathcal{F}\to\mathcal{F}$ such that, for every $f,\,g\in L^2(\R^3)$,
$$\mathfrak{R}_\omega^*(a(f)+a^*(\bar{g}))\mathfrak{R}_\omega=a(\omega\,f)+a^*(\overline{\omega\,g}).$$
The Shale-Stinespring condition (see \cite{Solovej}) ensures that, if $v$ is an Hilbert-Schmidt operator, then $\omega$ is indeed implementable. Thus $\mathfrak{R}_{\omega}$ is the implementor of the Bogoliubov transformation $\omega$. 

We consider the one-particle reduced density matrix associated with the Slater determinant \eqref{eq:slater-2q}:
$$\omega_N=\sum_{i=1}^N|f_i\rangle\langle f_i|,$$
the orthogonal projection on the $N$-dimensional space given by ${\rm Span}\{f_1,\dots,f_N\}$, where $\{f_i\}_{i=1,\dots,N}$ are the orbitals of the Slater determinant. Moreover, there exists a unitary operator $\mathfrak{R}_{\omega_N}$, implementor of a Bogoliubov transformation which generates a Slater determinant with orbitals $\{f_i\}_{i=1,\dots,N}$, namely
\be\label{eq:slater-fock}
\mathfrak{R}_{\omega_N}\Omega=a^*(f_1)\dots a^*(f_N)\,\Omega,
\ee
such that for every $g\in L^2(\R^3)$
\be\label{eq:property-slater}
\mathfrak{R}^*_{\omega_N}\,a(g)\,\mathfrak{R}_{\omega_N}=a(u_N\,g)+a^*(\bar{v}_N\,\bar{g})
\ee
where $ u_N=1-\omega_N$, $v_N=\sum_{i=1}^N|\bar{f_i}\rangle\langle f_i |$.
In other words, $\mathfrak{R}_{\omega_N}$ can be seen as the particle-hole transformation 
\begin{equation*}
\mathfrak{R}_{\omega_N}\,a^*(f_j)\,\mathfrak{R}^*_{\omega_N}=\left\{
\begin{array}{ll}
a(f_j), & j\leq N\\
a^*(f_j), & j>N
\end{array}
\right.
\end{equation*} 
where $\{f_j\}_{j=1,\dots,\infty}$ is the orthonormal basis of $L^2(\R^3)$ obtained by completing the orthonormal system $\{f_j\}_{j=1,\dots,N}$. 

Eq.ns \eqref{eq:slater-fock} and \eqref{eq:property-slater} are convenient representations of Slater determinants and the main reason to look at the Fock space in this context. Indeed, $\mathfrak{R}_{\omega_N}$ describes fluctuations in the Slater determinant with reduced density $\omega_N$: the Slater determinant is the new vacuum after the action of $\mathfrak{R}_{\omega_N}$; creation and annihilation operators act creating a particle outside the Slater determinant and annihilating  a particle inside the Slater determinant. 

%%%%%%%%%%%%%%%%%%%%%%%%%%%%%%%%%%%%%%%%%%%%%
%%%%%%%%%%%%%%%%%%%%%%%%%%%%%%%%%%%%%%%%%%%%%
\section{Sketch of the proof of Theorem \eqref{thm:main}}

The proof of Theorem \eqref{thm:main} is a direct consequence of the following Proposition (an adaptation of Theorem 2.2 in \cite{PRSS})

\begin{proposition}\label{prop:thm2}
Let $\omega_N$ be a sequence of orthogonal projections on $L^2 (\bR^3)$, with $\tr \, \omega_N = N$ and $\tr\, (-\eps^2 \Delta) \, \omega_N \leq C N$. Let $\omega_{N,t}$ denote the solution of the Hartree-Fock equation (\ref{eq:HF}) with initial data $\omega_{N,0} = \omega_N$.  We assume that there exist $T > 0$, $p > 6/(3-2\alpha-6\delta)$ for $\delta\in(0,\frac{1}{2})$ and $C > 0$ such that 
\begin{equation}\label{eq:assFS} \sup_{t \in [0;T]} \sum_{i=1}^3 \left[ \| \rho_{|[x_i, \omega_{N,t}]|} \|_1 +  \| \rho_{|[x_i,\omega_{N,t}]|} \|_p \right]  \leq C N\eps .
\end{equation}
Let $\xi_N \in \mathcal{F}$ be a sequence with 
\[ \langle \xi_N, \mathcal{N} \xi_N \rangle \leq C N^\beta \]  
for an exponent $\beta$, with $0 \leq \beta < 1$.  
We consider the evolution 
\[ \Psi_{N,t} = e^{-i\mathcal{H}_N t/\eps} R_{\omega_N} \xi_N \]
and denote by $\gamma^{(1)}_{N,t}$ the one-particle reduced density of $\Psi_{N,t}$, as defined in \eqref{eq:1pdm}. Then for all $\eta > 0$ there is a constant $C > 0$ such that 
\[ \sup_{t \in [0;T]} \left\| \gamma_{N,t}^{(1)} - \omega_{N,t} \right\|_\text{HS} \leq C \left[ N^{\beta/2} + N^{(3-2\alpha)/2(3-\alpha) + \eta} \right] \]
and 
\[ \sup_{t \in [0;T]} \tr \, \left| \gamma_{N,t}^{(1)} - \omega_{N,t} \right| \leq C \left[ N^{(1+\beta)/2} + N^{(3-2\alpha)/2(3-\alpha) + \eta}\right] \, .\]
\end{proposition}

The proofs of Theorem \ref{thm:main} and Proposition \ref{prop:thm2} are immediate adaptations of Theorems 1.1 and 2.2 in \cite{PRSS} and can be found in Section 2 of \cite{PRSS}.
 The proof of Proposition \ref{prop:thm2} relies on the control on the growth of fluctuations established in Proposition \ref{prop:thm3}. To prove Proposition \ref{prop:thm3}, we use a generalized Fefferman-de La Llave representation formula for the inverse power law potential $1/|x-y|^{\alpha}$ and the estimate, stated in the following Lemma, on the trace norm of the commutator between the multiplication operator given by the smooth function $\chi_{(r,z)}(x)=e^{{|x-z|^2}/{r^2}}$ and the solution $\omega_{N,t}$ to the Hartree-Fock equation \eqref{eq:HF}:
 \begin{lemma}[Lemma 3.1 in \cite{PRSS}]
\label{lemma:traccia}
Let $\chi_{r,z} (x) = \exp (-x^2/r^2)$. Then, for all $0 < \delta <1/2$ there exists $C > 0$ such that the pointwise bound
\begin{equation}\label{eq:tr1} \left\| [\chi_{(r,z)} , \omega_{N,t} ] \right\|_\text{tr} 
\leq C \, r^{\frac{3}{2} - 3\delta} \sum_{i=1}^3 \| \rho_{|[x_i, \omega_{N,t}]|} \|_1^{\frac{1}{6}+\delta} \left( \rho^*_{|[x_i, \omega_{N,t}]|} (z) \right)^{\frac{5}{6} - \delta}
\end{equation}
holds true. Here $\varrho^*_{|[x_i,\omega_{N,t}]|}$ denotes the Hardy-Littlewood maximal function defined by 
\begin{equation}\label{eq:max-def} \rho^*_{|[x_i , \omega_{N,t}]|} (z) = \sup_{B : z \in B} \frac{1}{|B|} \int_B dx\,\rho_{|[x_i , \omega_{N,t}]|} (x)  \end{equation}
with the supremum taken over all balls $B \in \bR^3$ such that $z \in B$. 
\end{lemma}
The proof of the above Lemma can be found in \cite{PRSS}.

%%%%%%%%%%%%%%%%%%%%%%%%%%%%%%%%%%%%%%%%%%%%%
%%%%%%%%%%%%%%%%%%%%%%%%%%%%%%%%%%%%%%%%%%%%%
\section{Control on the growth of fluctuations}

We consider $\xi_N\in\mathcal{F}$ a vector in the fermionic Fock space such that the number of particles in the state $\xi_N$ is bounded by a power of $N$, i.e.
$$\langle \xi_N,\mathcal{N}\xi_N \rangle\leq C\,N^\beta$$
for some $\beta\in[0,1)$ and a positive constant $C$. 
The time-evolution of $\xi_N$ in the Fock space is given by the action of the semigroup generated by the Hamiltonian \eqref{eq:H-FS}
$$\Psi_{N,t}=e^{-i\mathcal{H}_Nt/\e}\mathfrak{R}_{\omega_N}\xi_N.$$
In order to prove Proposition \ref{prop:thm2} (and thus Theorem \ref{thm:main}) we compare  the one-particle reduced density of $\Psi_{N,t}$ with the solution $\omega_{N,t}$ of the time-dependent Hartree-Fock equation \eqref{eq:HF}. To this end, we introduce the {\it fluctuation dynamics}
\be\label{eq:fluct}
\mathcal{U}_N(t)=\mathfrak{R}_{\omega_{N,t}}^*e^{-i\mathcal{H}_N t/\e}\mathfrak{R}_{\omega_N}
\ee 
so that 
$$\Psi_{N,t}=\mathfrak{R}_{\omega_{N,t}}\mathcal{U}_N(t)\xi_N,$$
where $\mathcal{U}_N(t)\xi_N\in\mathcal{F}$ describes the excitations of the Slater determinant at time $t>0$. The following Proposition ensures the boundedness of the expectation of the number of excitations of the Slater determinant in the state $\Psi_{N,t}$.

\begin{proposition}\label{prop:thm3}
Let $\omega_N$ be a fermionic operator such that $0\leq\omega_N\leq 1$, $\omega_N^2=\omega_N$ and $\tr\ \omega_N=N$. Assume that
\begin{itemize}
\item[i)] $\tr\ (-\e^2\Delta)\omega_N\leq CN$;
\item[ii)] there exist a time $T>0$ and a number $p>6/(3-2\alpha-6\delta)$ for some $\delta>0$, such that
$$\sup_{t\in[0,T]}\sum_{i=1}^3\left[ \|\rho_{|[x_i,\omega_{N,t}]|}\|_{L^1}+\|\rho_{|[x_i,\omega_{N,t}]|}\|_{L^{p}} \right]\leq C\,N\,\e$$
where $C$ is a positive constant. 
\end{itemize}  
Let $\mathcal{U}_N(t)$ be the fluctuation dynamics defined in \eqref{eq:fluct} and $\xi_N\in\mathcal{F}$, $\|\xi_N\|=1$. Then, for every $\eta>0$, there exists a positive constant $C$ such that 
\be\label{eq:est-fluct}
\sup_{t\in[0,T]}\langle \mathcal{U}_N(t)\,\xi_N,\mathcal{N}\,\mathcal{U}_N(t)\,\xi_N \rangle \leq C\left[ \langle \xi_N,\mathcal{N}\xi_N\rangle + N^{(3-2\alpha-6\delta)/(3-\alpha)} \right].
\ee
\end{proposition}

\begin{proof}
To bound \eqref{eq:est-fluct} we look for a Gr\"{o}nwall type estimate on the quantity which represents the expectation of excitations of the Slater determinant. 
We perform the time derivative of $\langle \mathcal{U}_N(t)\xi_N,\mathcal{N}\mathcal{U}_N(t)\xi_N\rangle$ and straightforward computations (see Proposition 3.3 in \cite{BPS13} for details) lead to
\be\label{eq:gronwall-fluct}
\begin{split}
i\,\e\,\partial_t\,\langle \mathcal{U}_N(t)\xi_N,&\mathcal{N}\mathcal{U}_N(t)\xi_N \rangle\\
=&\frac{4\,i}{N}\,\mathfrak{Im}\iint\frac{1}{|x-y|^\alpha}  \\
&\left\{ \langle \mathcal{U}_N(t)\xi_N,\,a^*(u_{t,x})\,a(\overline{v_{t,y}})\,a(u_{t,y})\,a(u_{t,x})\,\mathcal{U}_N(t)\xi_N \rangle  \right.\\
& \langle \mathcal{U}_N(t)\xi_N,\,a^*(u_{t,y})\,a^*(\overline{v_{t,y}})\,a^*(\overline{v_{t,x}})\,a(\overline{v_{t,x}})\,\mathcal{U}_N(t)\xi_N \rangle  \\
&\left. \langle \mathcal{U}_N(t)\xi_N,\,a(\overline{v_{t,x}})\,a(\overline{v_{t,y}})\,a(u_{t,y})\,a(u_{t,x})\,\mathcal{U}_N(t)\xi_N \rangle \right\}\,dx\,dy\\
=& I+II+III
\end{split}
\ee
where $u_{t,x}(z)=u_{N,t}(x;z)$, $v_{t,x}(z)=v_{N,t}(x;z)$,
\begin{equation*}
\begin{split}
I&=\frac{4\,i}{N}\,\mathfrak{Im}\iint\frac{1}{|x-y|^\alpha}\langle \mathcal{U}_N(t)\xi_N,\,a^*(u_{t,x})\,a(\overline{v_{t,y}})\,a(u_{t,y})\,a(u_{t,x})\,\mathcal{U}_N(t)\xi_N \rangle\,dx\,dy,\\
II&=  \frac{4\,i}{N}\,\mathfrak{Im}\iint\frac{1}{|x-y|^\alpha}  \langle \mathcal{U}_N(t)\xi_N,\,a^*(u_{t,y})\,a^*(\overline{v_{t,y}})\,a^*(\overline{v_{t,x}})\,a(\overline{v_{t,x}})\,\mathcal{U}_N(t)\xi_N \rangle\,dx\,dy,\\
III&= \frac{4\,i}{N}\,\mathfrak{Im}\iint\frac{1}{|x-y|^\alpha} \langle \mathcal{U}_N(t)\xi_N,\,a(\overline{v_{t,x}})\,a(\overline{v_{t,y}})\,a(u_{t,y})\,a(u_{t,x})\,\mathcal{U}_N(t)\xi_N \rangle  \,dx\,dy.
\end{split}
\end{equation*}
To bound the r.h.s. in \eqref{eq:gronwall-fluct} we make use of a smooth version of a generalization of Fefferman-de La Llave representation formula for radial potentials \cite{FDL}, \cite{HS}. In \cite{HS} an explicit expression for radial potentials which exhibit some decay at infinity is provided. In the context under consideration such a formula resumes in
\be\label{eq:FDLL}
\frac{1}{|x-y|^\alpha}=\frac{4}{\pi^2}\int_0^\infty \frac{dr}{r^{4+\alpha}}\int dz\,\chi_{(r,z)}(x)\,\chi_{(r,z)}(y)\,,
\ee 
where $\chi_{(r,z)}(x)=e^{-|x-z|^2/r^2}$. The advantage of this representation consists in the fact that the smooth part of the inverse power law potential represented by $\chi_{(r,z)}(\cdot)$ is decoupled from the singular part. In the formulation we are using, this representation is useful to isolate the commutator structure which is estimated in Lemma \ref{lemma:traccia} (see Eq.n \eqref{eq:B-bound} below).

Therefore, plugging \eqref{eq:FDLL} in $I$, we obtain
\begin{equation*}
\begin{split}
I&\leq \frac{C}{N}\int_0^\infty\frac{dr}{r^{4+\alpha}}\iiint \chi_{(r,z)}(x)\,\chi_{(r,z)}(y)\,\\
&\quad\quad\quad\quad\quad\quad\quad\quad\langle \mathcal{U}_N(t)\xi_N,a^*(u_{t,x})\,a(\overline{v}_{t,y})\,a(u_{t,y})\,a(u_{t,x})\,\mathcal{U}_N(t)\xi_N  \rangle\,dx\,dy\,dz\\
&=\frac{C}{N}\int_0^\infty\frac{dr}{r^{4+\alpha}}\iint \chi_{(r,z)}(x)\langle \mathcal{U}_N(t)\xi_N,a^*(u_{t,x})\,B_{r,z}\,a(u_{t,x})\,\mathcal{U}_N(t)\xi_N \rangle\,dx\,dz
\end{split}
\end{equation*}
where 
$$B_{r,z}=\int a(\overline{v_{t,y}})\chi_{(r,z)}(y)a(u_{t,y})\,dy=\iint (\overline{v}_{N,t}\chi_{(r,z)}u_{N,t})(s_1;s_2)\,a_{s_1}a_{s_2}\,ds_1\,ds_2.$$
Lemma \ref{lemma:2quant} and the fact that $u$ and $v$ are orthogonal yield
\be\label{eq:B-bound}\|B_{r,z}\|\leq 2\,\tr\,|\overline{v}_{N,t}\chi_{(r,z)}u_{N,t}|\leq 2\,\tr\,|[\chi_{(r,z)},\omega_{N,t}]|.\ee
Therefore, 
\begin{equation*}
I\leq \frac{C}{N}\int_0^\infty\frac{dr}{r^{4+\alpha}}\iint \chi_{(r,z)}(x)\,\tr\ |[\chi_{(r,z)},\omega_{N,t}]|\,\|a(u_{t,x})\,\mathcal{U}_N(t)\xi_N\|^2\,dx\,dz\,.
\end{equation*}
By Lemma \ref{lemma:traccia} we get
\begin{equation*}
\begin{split}
|I|&\leq C\frac{(N\e)^{\frac{1}{6}+\delta}}{N}\sum_{i=1}^3 \int_0^\infty \frac{dr}{r^{\frac{5}{2}+\alpha+3\delta}}\int b_{r,i}(x)\,\|a(u_{t,x})\mathcal{U}_{N}(t)\xi_N\|^2 \,dx\\
&\leq C\frac{(N\e)^{\frac{1}{6}+\delta}}{N}\sum_{i=1}^3 \int_0^\infty \frac{dr}{r^{\frac{5}{2}+\alpha+3\delta}}\langle \mathcal{U}_N(t)\xi_N,\,d\Gamma(u_{N,t}b_{r,i}(x)u_{N,t})\mathcal{U}_{N}(t)\xi_N\rangle
\end{split}
\end{equation*}
where $d\Gamma(u_{N,t}b_{r,i}(x)u_{N,t})$ is the second quantization of the operator $u_{N,t}b_{r,i}(x)u_{N,t}$ with integral kernel 
$$
(u_{N,t}b_{r,i}(x)u_{N,t})(s_2;s_2)=\int u_{N,t}(s_1;x)b_{r,i}(x)u_{N,t}(x;s_2)\,ds_1\,ds_2\,
$$
and $b_{r,i}$ is defined as
$$
b_{r,i}(x)=\int \chi_{(r,z)}(x)\left( \rho^*_{|[x_i,\omega_{N,t}]|}(z) \right)^{\frac{5}{6}-\delta}\,dz\,.
$$
Notice that $\|u_{N,t}\|\leq1$, thus Lemma \ref{lemma:2quant} yields
$$
|I|\leq C\frac{(N\e)^{\frac{1}{6}+\delta}}{N}\sum_{i=1}^3\int_0^\infty \frac{dr}{r^{\frac{5}{2}+\alpha+3\delta}}\|b_{r,i}\|_{L^\infty}\|\mathcal{N}^{1/2}\mathcal{U}_N(t)\xi_N\|^2\,.
$$
Hardy-Littlewood maximal inequality then implies
\be\label{eq:HLM}
\|b_{r,i}\|_{L^\infty}\leq r^{\frac{3}{p}}\|\rho^*_{|[x_i,\omega_{N,t}]|}\|^{\frac{5}{6}-\delta}_{L^{\left(\frac{5}{6}-\delta\right)q}}\leq C\,r^{\frac{3}{p}}\|\rho_{|[x_i,\omega_{N,t}]|} \|^{\frac{5}{6}-\delta}_{L^{\left(\frac{5}{6}-\delta\right)q}}
\ee
where $p,\ q$ are conjugated H\"{o}lder exponents coupled by the relation $\frac{1}{p}+\frac{1}{q}=1$, with the constraint $\frac{6}{5-6\delta}< q<\infty$ in order to ensure $q>1$ so that the last inequality in the r.s.h. of Eq.n \eqref{eq:HLM} holds true. 

Now, we split the integral in the $r$ variable into two parts: let $k>0$ be a fixed positive number, then
\begin{equation*}
\begin{split}
|I|\leq C\frac{(N\e)^{\frac{1}{6}+\delta}}{N}\sum_{i=1}^3&\ \left[\int_0^k \frac{dr}{r^{\frac{5}{2}+\alpha+3\delta-\frac{3}{p}}}\,\|\rho_{|[x_i,\omega_{N,t}]|} \|^{\frac{5}{6}-\delta}_{L^{\left(\frac{5}{6}-\delta\right)q}}  \|\mathcal{N}^{1/2}\mathcal{U}_N(t)\xi_N\|^2\right.\\
&+ \left.\int_k^\infty \frac{dr}{r^{\frac{5}{2}+\alpha+3\delta-\frac{3}{p'}}}\,\|\rho_{|[x_i,\omega_{N,t}]|} \|^{\frac{5}{6}-\delta}_{L^{\left(\frac{5}{6}-\delta\right)q'}} \|\mathcal{N}^{1/2}\mathcal{U}_N(t)\xi_N\|^2\right]
\end{split}
\end{equation*}
where $(p,q)$ and $(p',q')$ are chosen to guarantee integrability of the integral in the $r$ variable, namely
$$
p<\frac{6}{3+2\alpha+6\delta},\quad\quad q>\frac{6}{3-2\alpha-6\delta},
$$ 
$$
p'>\frac{6}{3+2\alpha+6\delta},\quad\quad q'<\frac{6}{3-2\alpha-6\delta}.
$$
With these choices, using hypothesis $ii)$ in Proposition \ref{prop:thm3} we obtain for every $t\in [0,T]$
\be\label{eq:bound-I}
|I|\leq C\e\|\mathcal{N}^{1/2}\mathcal{U}_N(t)\xi_N\|^2=C\e\langle \mathcal{U}_N(t)\xi_N,\,\mathcal{N}\mathcal{U}_N(t)\xi_N \rangle.
\ee

The second term on the r.h.s. of \eqref{eq:gronwall-fluct} can be handled analogously:
\begin{equation*}
\begin{split}
II&\leq \frac{C}{N}\iint \frac{1}{|x-y|^\alpha}\langle \mathcal{U}_N(t)\xi_N,a^*(\overline{v}_{t,x})\,a^*(u_{t,x})\,a^*(\overline{v}_{t,y})\,a(\overline{v}_{t,x})\mathcal{U}_N(t)\xi_N \rangle\,dx\,dy\\
&=\frac{C}{N}\int_0^\infty \frac{dr}{r^{4+\alpha}}\iiint \chi_{(r,z)}(x)\chi_{(r,z)}(y)\\
&\quad\quad\quad\quad\quad\quad\quad\langle \mathcal{U}_N(t)\xi_N,a^*(\overline{v}_{t,x})\,a^*(u_{t,x})\,a^*(\overline{v}_{t,y})\,a(\overline{v}_{t,x})\mathcal{U}_N(t)\xi_N \rangle\,dz\,dx\,dy\\
&=\frac{C}{N}\int dx \int_0^\infty \frac{dr}{r^{4+\alpha}}\int dz \chi_{(r,z)}(x)\langle \mathcal{U}_N(t)\xi_N,a^*(\overline{v}_{t,x})\,B^*_{r,z}\,a(\overline{v}_{t,x})\mathcal{U}_N(t)\xi_N \rangle\\
&\leq C\frac{(N\e)^{\frac{1}{6}+\delta}}{N}\sum_{i=1}^3\int_0^\infty\frac{dr}{r^{\frac{5}{2}+\alpha+3\delta}}\iint dx\,dz\,\chi_{(r,z)}(x)\,\left( \rho^*_{|[x_i,\omega_{N,t}]|}(z)\right)^{\frac{5}{6}-\delta}\\
&\quad\quad\quad\quad\quad\quad\quad\quad\quad\quad\quad\quad\quad\quad\quad\quad\quad\quad\quad\quad\quad\quad\quad\quad\quad\|a(\overline{v}_{t,x})\mathcal{U}_N(t)\xi_N\|^2\\
&\leq C\frac{(N\e)^{\frac{1}{6}+\delta}}{N}\sum_{i=1}^3\int_0^\infty \frac{dr}{r^{\frac{5}{2}+\alpha|3\delta}}\langle \mathcal{U}_N(t)\xi_N,\,d\Gamma(\overline{v}_{N,t}\,b_{r,i}(x)\,\overline{v}_{N,t})\mathcal{U}_N(t)\xi_N \rangle\\
&\leq C\frac{(N\e)^{\frac{1}{6}+\delta}}{N}\sum_{i=1}^3 \int_0^\infty \frac{dr}{r^{\frac{5}{2}+\alpha+3\delta}}\|b_{r,i}\|_{L^\infty}\|\mathcal{N}^{1/2}\mathcal{U}_N(t)\xi_N\|^2 
\end{split}
\end{equation*}
and we conclude as in \eqref{eq:bound-I}
\be\label{eq:bound-II}
|II|\leq C\e\|\mathcal{N}^{1/2}\mathcal{U}_N(t)\xi_N\|^2=C\e\langle \mathcal{U}_N(t)\xi_N,\,\mathcal{N}\mathcal{U}_N(t)\xi_N \rangle.
\ee
In order to close the Gr\"{o}nwall estimate we need to bound $III$. This term, together with the initial quantity $\langle\xi_N,\mathcal{N}\xi_N\rangle$, determines the function of $N$ which bounds the expectation of the number of fluctuations.
To deal with $III$, we use again Eq.n \eqref{eq:FDLL} to get
\begin{equation*}
\begin{split}
III&\leq \frac{C}{N}\int_0^\infty \frac{dr}{r^{4+\alpha}}\iiint \chi_{(r,z)}(x)\chi_{(r,z)}(y)\\
&\quad\quad\quad\quad\quad\quad\quad\quad\quad\langle\mathcal{U}_N(t)\xi_N, a(\overline{v}_{t,x})\,a(\overline{v}_{t,y})\,a(u_{t,y})\,a(u_{t,x})\mathcal{U}_N(t)\xi_N\rangle\,dz\,dx\,dy\,.
\end{split}
\end{equation*}
We fix $k>0$, to be chosen later as a $\e$ dependent function, and divide $III$ into two parts: 
\begin{equation*}
\begin{split}
III_1&= \frac{C}{N}\int_0^k\frac{dr}{r^{4+\alpha}}\iiint \chi_{(r,z)}(x)\chi_{(r,z)}(y)\\
&\quad\quad\quad\quad\quad\quad\quad\quad\quad\langle \mathcal{U}_N(t)\xi_N,a(\overline{v}_{t,x})\,a(\overline{v}_{t,y})\,a(u_{t,y})\,a(u_{t,x})\mathcal{U}_N(t)\xi_N\rangle\,dz\,dx\,dy  \\
III_2&= \frac{C}{N}\int_k^\infty\frac{dr}{r^{4+\alpha}}\iiint \chi_{(r,z)}(x)\chi_{(r,z)}(y)\\
&\quad\quad\quad\quad\quad\quad\quad\quad\quad\langle \mathcal{U}_N(t)\xi_N,a(\overline{v}_{t,x})\,a(\overline{v}_{t,y})\,a(u_{t,y})\,a(u_{t,x})\mathcal{U}_N(t)\xi_N\rangle\,dz\,dx\,dy.
\end{split}
\end{equation*}
As for the term $III_1$, we observe that the integral in the $z$ variable chancels part of the singularity in the $r$ variable by producing a factor $r^3$, i.e.
\be\label{eq:z-int}\int \chi_{(r,z)}(x)\chi_{(r,z)}(y)\,dz=r^3\chi_{(\sqrt{2}r,x)}(y).\ee
Plugging \eqref{eq:z-int} into $III_1$ we obtain
\begin{equation*}
\begin{split}
III_1&=\frac{C}{N}\int_0^k\frac{dr}{r^{1+\alpha}}\int \langle \mathcal{U}_N(t)\xi_N, B_{\sqrt{2}r,x}\,a(\overline{v}_{t,x})\,a(u_{t,x})\mathcal{U}_N(t)\xi_N\rangle\,dx\\
&\leq\frac{C}{N}\int_0^k\frac{dr}{r^{1+\alpha}}\int \rho_t(x)^{1/2}\|B_{\sqrt{2}r,x}\|\,\|a(u_{t,x})\mathcal{U}_N(t)\xi_N\|\,dx\\
&\leq \frac{C}{N}\int_0^k\frac{dr}{r^{1+\alpha}}\int \rho_t(x)^{1/2}\tr\ |[\chi_{(\sqrt{2}r,x)},\omega_{N,t}]|\,\|a(u_{t,x})\mathcal{U}_N(t)\xi_N\|\,dx,
\end{split}
\end{equation*}
where we have used  
$$\|\overline{v}_{N,t}\|^2=\omega_{N,t}(x;x)=\rho_t(x)$$
in the first inequality
and \eqref{eq:B-bound} in the last inequality. Lemma \ref{lemma:traccia}, together with hypothesis $ii)$,  yields
 \be\label{eq:III-1}
 \begin{split}
 |III_1|&\leq C\frac{(N\e)^{\frac{1}{6}+\delta}}{N}\sum_{i=1}^3\int_0^k dr\,r^{\frac{1}{2}-\alpha-3\delta}\int dx\,\rho_t(x)^{1/2}\sum_{i=1}^3 \|\rho_{|[x_i,\omega_{N,t}]|}\|_{L^1}^{\frac{1}{6}+\delta}\left(\rho^*_{|[x_i,\omega_{N,t}]|}(x) \right)^{\frac{5}{6}-\delta}\\
 &\leq C\frac{(N\e)^{\frac{1}{6}+\delta}}{N}k^{\frac{3}{2}-\alpha-3\delta}\int \rho_t(x)^{1/2}\left(\rho^*_{|[x_i,\omega_{N,t}]|}(x)\right)^{\frac{5}{6}-\delta}\|a(u_{t,x})\mathcal{U}_N(t)\xi_N\|\,dx
 \end{split}
 \ee
 H\"{o}lder inequality and hypothesis $ii)$ imply
 \begin{equation*}
 \begin{split}
 |III_1|&\leq C\frac{(N\e)^{\frac{1}{6}+\delta}}{N}k^{\frac{3}{2}-\alpha-3\delta}\|\rho_t\|^{\frac{1}{2}}_{L^{\frac{5}{3}}}\sum_{i=1}^3\|\rho^*_{|[x_i,\omega_{N,t}]|}\|_{L^{\frac{25}{6}-5\delta}}^{\frac{5}{6}-\delta}\left(\int \|a(u_{t,x})\mathcal{U}_N(t)\xi_N\|^2\,dx\right)^{\frac{1}{2}}\\
 &\leq C\frac{(N\e)}{N}k^{\frac{3}{2}-\alpha-3\delta}\|\rho_t\|^{\frac{1}{2}}_{L^{\frac{5}{3}}}\langle \mathcal{U}_N(t)\xi_N,d\Gamma(u_{N,t})\mathcal{U}_N(t)\xi_N\rangle\\
 &\leq C\,\e\,k^{\frac{3}{2}-\alpha-3\delta}\|\rho_t\|^{\frac{1}{2}}_{L^{\frac{5}{3}}}\|\mathcal{N}^{1/2}\mathcal{U}_N(t)\xi_N\|^2\,.
 \end{split}
 \end{equation*}
Moreover, the $L^{5/3}$ norm of $\rho_t$ can be bounded in terms of the initial data by standard kinetic energy inequality that we report in  Appendix \ref{app:kinetic} for completeness. 

Thus by hypothesis $i)$, for every $t\in[0,T]$,
\be\label{eq:bound-III1}
\begin{split}
|III_1|&\leq C\,\sqrt{N}\,\e\,k^{\frac{3}{2}-\alpha-3\delta}\|\mathcal{N}^{1/2}\mathcal{U}_N(t)\xi_N\|\\
&\leq \e\|\mathcal{N}^{1/2}\mathcal{U}_N(t)\xi_N\|^2+C\,N\,\e\,k^{3-2\alpha-6\delta}.
\end{split}
\ee
  
As for the second term in $III$, we notice that 
\begin{equation*}
\begin{split}
|III_2|&\leq \frac{C}{N}\int_k^\infty\frac{dr}{r^{4+\alpha}}\int \|B_{r,z}\|^2\,dz\\
&\leq \frac{C}{N}\int_k^\infty\frac{dr}{r^{4+\alpha}}\int \tr\ |[\chi_{(r,z)},\omega_{N,t}]|\,dz\\
&\leq C\frac{(N\e)^2}{N}\int_k^\infty\frac{dr}{r^{1+\alpha+6\delta}}
\end{split}
\end{equation*}
The term on the r.h.s. is integrable being $\alpha> 0$
 and $\delta\in(0,1/2)$, thus
 \be\label{eq:bound-III2}
 |III_2|\leq C\,N\,\e^2\,k^{-\alpha-6\delta}.
 \ee
Combining the bounds \eqref{eq:bound-III1} and \eqref{eq:bound-III2}, we obtain
$$
|III|\leq \e\|\mathcal{N}^{1/2}\mathcal{U}_N(t)\xi_N\|^2+C\,N\,\e\,k^{3-2\alpha-6\delta}+C\,N\,\e^2\,k^{-\alpha-6\delta}.
$$
By optimizing in $k$, we get 
$$
k=\e^{1/(3-\alpha)}
$$
and therefore $III$ is bounded by
\be\label{eq:bound-III}
|III_3|\leq \e\|\mathcal{N}^{1/2}\mathcal{U}_N(t)\xi_N\|^2+C\,N\,\e^{3(2-\alpha-2\delta)/(3-\alpha)}
\ee
Eq.ns \eqref{eq:bound-I}, \eqref{eq:bound-II} and \eqref{eq:bound-III} lead to a control on the growth of fluctuations quantified by the following Gr\"{o}nwall inequality
$$
\left| \frac{d}{dt}\langle \mathcal{U}_N(t)\xi_N,\mathcal{N}\mathcal{U}_N(t)\xi_N \rangle \right| \leq C\langle \mathcal{U}_N(t)\xi_N,\mathcal{N}\mathcal{U}_N(t)\xi_N \rangle+ C\,N\,\e^{(3-2\alpha-6\delta)/(3-\alpha)}
$$
that implies, for every $t\in[0,T]$,
$$
\langle \mathcal{U}_N(t)\xi_N,\mathcal{N}\mathcal{U}_N(t)\xi_N\rangle\leq C\left[ \langle\xi_N,\mathcal{N}\xi_N\rangle + N\e^{(3-2\alpha-6\delta)/(3-\alpha)} \right]
$$
and the Proposition is proved. 
\end{proof}

%%%%%%%%%%%%%%%%%%%%%%%%%%%%%%%%%%%%%%%%%%%%%
%%%%%%%%%%%%%%%%%%%%%%%%%%%%%%%%%%%%%%%%%\

\appendix
\section{Kinetic energy estimates}\label{app:kinetic}

To bound the $L^{5/3}$ norm of the density $\rho_t$, we observe that Lieb-Thirring inequality and the positivity of the interaction potential yield
$$\|\rho_t\|_{L^{5/3}}^{5/3}\leq \tr\ (-\Delta)\omega_{N,t}\leq \e^{-2}\mathcal{E}_{\rm HF}(\omega_{N,t})$$
where $\mathcal{E}_{\rm HF}$
is the Hartree-Fock energy functional defined in \eqref{eq:E-HF}. Conservation of energy implies 
$$\|\rho_t\|_{L^{5/3}}^{5/3}\leq \e^{-2}\mathcal{E}_{\rm HF}(\omega_{N,t})=\e^{-2}\mathcal{E}_{\rm HF}(\omega_{N}).$$
To close the estimate using the assumption on the kinetic energy of the initial sequence $\tr\ (-\e^2\Delta)\omega_{N}\leq C\,N$, we observe that the potential energy can be bounded by the kinetic energy. Indeed, Hardy-Littlewood-Sobolev inequality yiels
\begin{equation*}
\frac{1}{N}\int \frac{1}{|x-y|^\alpha}\rho_0(x)\rho_0(y)\,dx\,dy\leq \frac{C}{N}\|\rho_0\|^2_{L^{6/(5-\alpha)}}
\end{equation*}
By interpolation, using that $\frac{6}{6-\alpha}\in\left(1,\frac{5}{3}\right)$ and  $\|\rho_0\|_{L^1}=N$, we have
\begin{equation*}
\begin{split}
\frac{1}{N}\int \frac{1}{|x-y|^\alpha}\rho_0(x)\rho_0(y)\,dx\,dy&\leq\frac{C}{N}\|\rho_0\|_{L^1}^{\frac{12-5\alpha}{6}}\|\rho_0\|_{L^{\frac{5}{3}}}^{\frac{5}{6}\alpha} \\
&\leq C\,N^{1-\frac{5}{6}\alpha}\|\rho_0\|_{L^{\frac{5}{3}}}^{\frac{5}{6}\alpha}\\
&\leq C\,N+N^{-\frac{2}{3}\alpha}\|\rho_0\|^{\frac{5}{3}\alpha}_{L^{\frac{5}{3}}},
\end{split}
\end{equation*}
where in the last line we have used Young inequality 
$ab\leq \frac{a^p}{p}+\frac{b^q}{q}$, $p^{-1}+q^{-1}=1$,
on the quantities $a=N^{1-\frac{\alpha}{2}}$ and $b=N^{-\frac{\alpha}{3}}\|\rho_0\|_{L^{\frac{5}{3}}}^{\frac{5}{6}\alpha}$ with $p=2/(2-\alpha)$ and $q=2/\alpha$. Thus, applying again Lieb-Thirring inequality and recalling that $\e=N^{-1/3}$, we obtain
\begin{equation*}
\frac{1}{N}\int \frac{1}{|x-y|^\alpha}\rho_0(x)\rho_0(y)\,dx\,dy\leq C\,N+C\,\tr\ (-\e^2\Delta)\omega_N
\end{equation*}
that is bounded by
\begin{equation*}
\frac{1}{N}\int \frac{1}{|x-y|^\alpha}\rho_0(x)\rho_0(y)\,dx\,dy\leq C\,N
\end{equation*}
thanks to assumption $i)$. 
\bigskip

{\bf Acknowledgement.}
The author is supported by the grant SNSF Ambizione S-71119-02-01.

\end{document}